\documentclass[%
 reprint,
 superscriptaddress,
 amsmath,amssymb,
 aps,
 amsthm,
 pra,
]{revtex4-1}

\usepackage{graphicx}
\usepackage{dcolumn}
\usepackage{bm}
\usepackage{braket}%
\usepackage{theorem}
\usepackage{enumerate}
\usepackage{algorithmic}

\usepackage[dvipdfm,usenames]{color}
\usepackage[dvipdfmx,colorlinks,citecolor=blue,linkcolor=blue,urlcolor=blue,bookmarksopenlevel=4]{hyperref}

\usepackage{color}

\theorembodyfont{\upshape}
\newtheorem{thm}{Theorem}
\newtheorem{proposition}[thm]{Proposition}
\newtheorem{lemma}[thm]{Lemma}
\newtheorem{cor}[thm]{Corollary}

\newtheorem{proof}{Proof}

\newtheorem{proofa}{Proof of Theorem~\ref{thm:LOCC2}}

\newcommand{\mH}{\mathcal{H}}

\newcommand{\mL}{\mathcal{L}}

\newcommand{\rhoE}{\tilde{\rho}}
\newcommand{\PiE}{\tilde{\Pi}}
\newcommand{\PE}{\tilde{P}}
\newcommand{\PhiE}[1]{\tilde{\Phi}^{(#1)}}
\newcommand{\OmegaE}{\tilde{\Omega}}

\newcommand{\mHE}{\tilde{\mH}}

\newcommand{\ZE}{\tilde{Z}}
\newcommand{\cE}{\tilde{c}}

\newcommand{\spn}{{\rm span}}

\newcommand{\QED}{\hspace*{0pt}\hfill $\blacksquare$}

\newcommand{\Tr}{{\rm Tr}}

\begin{document}

\preprint{APS/123-QED}

\title{Optimal Discrimination of Quantum States on a Two-Dimensional Hilbert Space
by Local Operations and Classical Communication}%

\affiliation{%
 Yokohama Research Laboratory, Hitachi, Ltd.,
 Yokohama, Kanagawa 244-0817, Japan
}%
\affiliation{
 School of Information Science and Technology,
 Aichi Prefectural University,
 Nagakute, Aichi 480-1198, Japan
}%
\affiliation{%
 Quantum Information Science Research Center, Quantum ICT Research Institute,
 Tamagawa University, Machida, Tokyo 194-8610, Japan
}%

\author{Kenji Nakahira}
\affiliation{%
 Yokohama Research Laboratory, Hitachi, Ltd.,
 Yokohama, Kanagawa 244-0817, Japan
}%
\affiliation{%
 Quantum Information Science Research Center, Quantum ICT Research Institute,
 Tamagawa University, Machida, Tokyo 194-8610, Japan
}%

\author{Tsuyoshi \surname{Sasaki Usuda}}
\affiliation{
 School of Information Science and Technology,
 Aichi Prefectural University,
 Nagakute, Aichi 480-1198, Japan
}%
\affiliation{%
 Quantum Information Science Research Center, Quantum ICT Research Institute,
 Tamagawa University, Machida, Tokyo 194-8610, Japan
}%

\date{\today}

\begin{abstract}
 We study the discrimination of multipartite quantum states
 by local operations and classical communication.
 We derive that any optimal discrimination of quantum states spanning a two-dimensional Hilbert space
 in which each party's space is finite dimensional
 is possible by local operations and one-way classical communication,
 regardless of the optimality criterion used and how entangled the states are.
\end{abstract}

\maketitle

\section{Introduction}

If two or more physically separated parties cannot communicate quantum information,
their possibilities of measuring quantum states are severely restricted.
Intuitively, product states seem to be able to be optimally distinguished using only
local operations and classical communication (LOCC),
while entangled states seem to be indistinguishable.
However, Bennett {\it et al.} found that orthogonal pure product states exist
that cannot be perfectly distinguished by LOCC \cite{Ben-DiV-Fuc-Mor-Rai-Cho-Smo-Woo-1999}.
Later, Walgate {\it et al.} proved that any two pure orthogonal states in finite-dimensional systems
can be distinguished with certainty
using local operations and one-way classical communication (one-way LOCC)
no matter how entangled they are \cite{Wal-Sho-Har-Ved-2000}.
These results encourage further investigations on the distinguishability of quantum states
by LOCC, and several important results have been reported in the case of orthogonal states
\cite{Gho-Kar-Roy-Sen-Sen-2001,Hor-Sen-Sen-Hor-2003,Fan-2004,Hay-Mar-Mur-Owa-Vir-2006,Vir-Sac-Ple-Mar-2001,Oga-2006}.
In this paper, we consider only finite-dimensional systems.

The problem of LOCC discrimination for non-orthogonal states
is much more complicated.
One of the main reasons is that perfect discrimination between them is impossible,
even without LOCC restriction.
Instead, optimal discrimination can be sought.
Walgate {\it et al.} \cite{Wal-Sho-Har-Ved-2000} posed the question:
``Can any non-orthogonal states on a two-dimensional (2D) Hilbert space be optimally distinguished by LOCC?''
To definitively answer this question, we must consider all optimality criteria.
Various optimality criteria have been suggested,
such as the Bayesian criterion, the Neyman-Pearson criterion, and the mutual information criterion,
but the above question is not answered except for very special cases,
such as an optimal error-free measurement for two non-orthogonal pure states
\cite{Che-Yan-2002,Ji-Cao-Yin-2005}.
Another reason is that optimal discrimination for non-orthogonal states often requires
a non-projective measurement on the space spanned by the given states,
while any orthogonal states can be perfectly distinguished by projective measurement.
A positive operator-valued measure (POVM) is the most general formulation of a measurement
permitted by quantum mechanics and is commonly adopted in quantum information theory \cite{Nie-Chu-2000}.
We denote a measurement on a 2D Hilbert space as a 2D measurement.
Some important examples of 2D non-projective measurements are a measurement
maximizing the success rate for more than two states on a 2D Hilbert space
and a measurement giving the result ``don't know'' with non-zero probability,
such as an inconclusive measurement
\cite{Che-Bar-1998,Eld-2003-inc,Fiu-Jez-2003}.

Let $\mH_{\rm ex}$ be a composite Hilbert space and $\mH_{\rm sub}$ be a subspace of $\mH_{\rm ex}$.
For simplicity, we say that a measurement described by the POVM $\{ \Pi_m \}$ on $\mH_{\rm sub}$
can be realized by LOCC (or one-way LOCC) if there exists an LOCC measurement (or a one-way LOCC measurement)
described by the POVM $\{ E_m \}$
on $\mH_{\rm ex}$ such that $\Pi_m = P_{\rm sub} E_m P_{\rm sub}$ for any index $m$,
where $P_{\rm sub}$ is the orthogonal projection operator onto $\mH_{\rm sub}$.
If any measurement on $\mH_{\rm sub}$ can be realized by LOCC, then
any quantum states on $\mH_{\rm sub}$ can be optimally distinguished using only LOCC.
Walgate {\it et al.}'s question can be rephrased as
``Can any measurement on a 2D Hilbert space be realized by LOCC?''

We emphasize that this question would be quite difficult to answer.
Instead of a 2D non-projective measurement, one might consider realizing a corresponding
projective measurement, which is obtained by Naimark's theorem \cite{Neu-1940}, by LOCC.
According to Naimark's theorem, any non-projective measurement can be realized by a projective measurement
on an extended Hilbert space.
However, if a 2D non-projective measurement has more than two POVM operators,
then so does the corresponding projective measurement,
and such a measurement often cannot be realized by LOCC
\cite{Wal-Har-2002,Dua-Fen-Xin-Yin-2009,Chi-Dua-Hsi-2014}.
Thus, this approach cannot directly answer the question.
Alternatively, one might try to decompose a given 2D non-projective measurement
into several 2D projective measurements.
It is known that there exist ``decomposable'' measurements,
which statistically give the same results as randomly choosing among measurements
each of which has fewer POVM operators than the original one \cite{Dar-Pre-Per-2005}.
If a 2D measurement can be decomposed into 2D projective measurements,
then from \cite{Wal-Sho-Har-Ved-2000}, it can obviously be realized by LOCC.
However, only a few 2D non-projective measurements are decomposable \cite{Dar-Pre-Per-2005}.

In this paper, we show that any 2D measurement can be realized by one-way LOCC
no matter how many POVM operators it has.
Our result answers the above question:
A global measurement is not needed for a 2D measurement in finite-dimensional systems,
regardless of the optimality criterion used.

It is worth noting that the problem of realizing a measurement by one-way LOCC
is closely related to realizing a quantum receiver.
Realization of an optimal or suboptimal receiver for optical states using linear optical feedback
(or feedforward) and photon counting has been widely studied both theoretically and experimentally
\cite{Dol-1973,Coo-2007,Tak-Sas-Loo-Lut-2005,Tak-Sas-Lut-2006,Nak-Usu-2012-receiver,Nak-Usu-2014,Wit-And-Tak-Syc-Leu-2010,Guh-Hab-Tak-2010,Ass-Poz-Pie-2011,Mul-Usu-Wit-Tak-Mar-And-Leu-2012}.
This type of receiver performs an individual measurement on each temporal or spatial slot.
A measurement can be decomposed into such individual measurements
if it can be realized by one-way LOCC;
thus, our result indicates that any 2D measurement can be decomposed into individual measurements,
at least in finite-dimensional systems.
It is often important to investigate whether a measurement can be realized by one-way LOCC
to check whether it can be implemented using only feasible resources
when the whole system is spatially or temporally separated.

In Section~\ref{sec:preliminaries}, we present some necessary preliminaries, where
we show that any 2D measurement can be realized by one-way LOCC
if any measurement with finite rank-one POVM operators on any 2D bipartite Hilbert space
in which Alice's subspace is two-dimensional can be realized by one-way LOCC.
In Section~\ref{sec:proj}, we recall the idea of Walgate {\it et al.} \cite{Wal-Sho-Har-Ved-2000},
which provides a method for realizing a 2D projective measurement by one-way LOCC.
In Section~\ref{sec:nonproj}, we consider realizing a 2D non-projective measurement by one-way LOCC.
We show that, by extending Walgate {\it et al.}'s idea, any measurement with finite rank-one POVM operators
on any 2D bipartite Hilbert space
in which Alice's subspace is two-dimensional can be realized by one-way LOCC
(Propositions~\ref{pro:LOCC1} and \ref{pro:LOCC2}; also Theorem~\ref{thm:LOCC2}).
We conclude the paper in Section~\ref{sec:conclusion}.

\section{Preliminaries} \label{sec:preliminaries}

We first consider a bipartite system.
Let $\ket{\psi}$ and $\ket{\phi}$ be two linearly independent quantum states
shared by Alice and Bob.
We can write, in general form,
\begin{eqnarray}
 \ket{\psi} &=& \sum_n \ket{p_n}_A \ket{q_n}_B, \nonumber \\
 \ket{\phi} &=& \sum_n \ket{p_n}_A \ket{r_n}_B, \label{eq:psi_phi}
\end{eqnarray}
where $\{ \ket{p_n}_A \}$ are quantum states of Alice, and
$\{ \ket{q_n}_B \}$ and $\{ \ket{r_n}_B \}$ are quantum states of Bob.
$\{ \ket{p_n}_A \}$, $\{ \ket{q_n}_B \}$, and $\{ \ket{r_n}_B \}$ are
generally unnormalized and non-orthogonal.
Let $\mH_A = \spn(\{ \ket{p_n}_A \})$ and $\mH_B = \spn(\{ \ket{q_n}_B \}, \{ \ket{r_n}_B \})$.
Also, let $\mH$ be a 2D Hilbert space spanned by $\ket{\psi}$ and $\ket{\phi}$.
We denote such $\mH$ as a 2D $(N_A, N_B)$-space,
where $N_A = \dim~\mH_A$ and $N_B = \dim~\mH_B$.
We consider finite-dimensional systems; $N_A$ and $N_B$ are finite.
Assume that Alice and Bob share one of a known collection of $L$ quantum states
represented by density operators $\{ \rho_l \}_{l=1}^L$
on $\mH$ and want to optimally discriminate between them in a certain optimality criterion.
Our main result is that
any 2D measurement (in finite-dimensional systems) can be realized by one-way LOCC
(see Corollary~\ref{cor:LOCC}),
which indicates that any optimal discrimination of $\{ \rho_l \}$
can be realized by one-way LOCC.

We can easily extend our result to multipartite systems
in a way similar to \cite{Wal-Sho-Har-Ved-2000}.
Here, let us imagine a tripartite system:
Alice, Bob, and Charlie share two linearly independent quantum states
$\ket{\psi}$ and $\ket{\phi}$,
which can be represented by
\begin{eqnarray}
 \ket{\psi} &=& \sum_n \ket{p_n'}_A \ket{q_n'}_{BC}, \nonumber \\
 \ket{\phi} &=& \sum_n \ket{p_n'}_A \ket{r_n'}_{BC}, \label{eq:psi_phi_multi}
\end{eqnarray}
instead of (\ref{eq:psi_phi}).
In (\ref{eq:psi_phi_multi}), Bob and Charlie are first grouped as one party.
Asuume that our main result, i.e., Corollary~\ref{cor:LOCC}, holds in a bipartite system;
then, we can show that any measurement on any tripartite 2D Hilbert space can also be realized by one-way LOCC.
Indeed, Alice performs a measurement on her system according to the bipartite one-way LOCC protocol
that we will propose in this paper and tells the result to Bob and Charlie.
Then, Bob and Charlie can again use the same protocol.
This argument can easily be extended to any multipartite system,
and thus, in the rest of paper, we consider only bipartite systems.

First, we show that our main reulst can be reduced to a simpler one.
For example, from \cite{Chi-Dar-Sch-2007},
any quantum measurement with a continuous set of outcomes (including the discrete outcomes)
on a finite-dimensional Hilbert space is equivalent to a continuous random choice of measurements
with finite outcomes.
Thus, it suffices to show that any 2D measurement with finite outcomes can be
realized by one-way LOCC.
We show the following lemma:
\begin{lemma} \label{lemma:2xM}
 If any measurement with finite rank-one POVM operators on a 2D $(2,N)$-space
 ($N$ is finite integer) can be realized by one-way LOCC,
 then any 2D measurement (in finite-dimensional systems) can be realized by one-way LOCC.
\end{lemma}

\begin{proof}
 Assume that any measurement with finite rank-one POVM operators on a 2D $(2,N)$-space
 (denoted by $\mH_2$) can be realized by one-way LOCC.

 First, we show that any measurement on $\mH_2$ can be realized by one-way LOCC.
 From \cite{Chi-Dar-Sch-2007},
 any quantum measurement, even if with a continuous set of outcomes, on $\mH_2$
 can always be realized as a random choice of extremal measurements on $\mH_2$,
 where an extremal measurement is an extremal point of the set of all possible POVMs,
 which is a convex set.
 Moreover, from \cite{Dar-Pre-Per-2005}, an extremal measurement on $\mH_2$ must be made of
 finite rank-one POVM operators, apart from the trivial POVM $\{ \Pi_1 = I_{\mH_2} \}$
 ($I_{\mH_2}$ is the identity operator on $\mH_2$).
 The trivial POVM is obviously realized by one-way LOCC;
 thus, we consider only a nontrivial POVM.

 Next, we show that a 2D measurement $\{ \Pi_m \}$ (on a 2D $(N_A,N_B)$-space)
 can be realized by one-way LOCC.
 Let $\mH_A$ and $\mH_B$ be Alice's and Bob's Hilbert spaces, respectively.
 The case of $N_A \le 2$ is trivial; assume that $N_A > 2$.
 Suppose without loss of generality that $N_A$ is even; otherwise,
 expand Alice's system into a $(N_A + 1)$-dimensional Hilbert space.
 Alice's system can be represented on the tensor product of two- and $(N_A/2)$-dimensional
 Hilbert spaces, denoted as $\mH_{A1}$ and $\mH_{A2}$, respectively.
 Since $\mH_A \otimes \mH_B = \mH_{A1} \otimes (\mH_{A2} \otimes \mH_B)$
 and $\dim~\mH_{A1} = 2$,
 $\{ \Pi_m \}$ can be realized by one-way LOCC between $\mH_{A1}$ and $\mH_{A2} \otimes \mH_B$.
 Thus, it suffices to show that a measurement on a 2D subspace of $\mH_{A2} \otimes \mH_B$
 can be realized by one-way LOCC.
 By repeating this procedure, the problem of realizing $\{ \Pi_m \}$ by one-way LOCC
 is reduced to the problem of realizing a measurement
 on a 2D $(2,N)$-space.
 Therefore, by the assumption, $\{ \Pi_m \}$ can be realized by one-way LOCC.
 \QED
\end{proof}

In this paper, we will prove the following theorem:
\begin{thm} \label{thm:LOCC2}
 Any measurement with finite rank-one POVM operators on a 2D $(2,N)$-space
 can be realized by one-way LOCC.
\end{thm}
From Lemma~\ref{lemma:2xM} and Theorem~\ref{thm:LOCC2},
we can easily obtain the following corollary (proof omitted):
\begin{cor} \label{cor:LOCC}
 Any 2D measurement (in finite-dimensional systems) can be realized by one-way LOCC.
\end{cor}


\section{Realization of 2D projective measurement by one-way LOCC} \label{sec:proj}

In this section, using an example, we recall the idea of Walgate {\it et al.} \cite{Wal-Sho-Har-Ved-2000},
which provides a way to realize a 2D projective measurement by one-way LOCC.
Let $\ket{\psi} = \ket{S}$ and $\ket{\phi} = \ket{T_0}$, where
\begin{eqnarray}
 \ket{S} &=& \frac{\ket{+}_A \ket{-}_B - \ket{-}_A \ket{+}_B}{\sqrt{2}}, \nonumber \\
 \ket{T_0} &=& \frac{\ket{+}_A \ket{-}_B + \ket{-}_A \ket{+}_B}{\sqrt{2}},
\end{eqnarray}
and $\{ \ket{+}_\alpha, \ket{-}_\alpha \}$ $~(\alpha \in \{ A, B \})$ is an orthonormal basis (ONB) in $\mH_\alpha$.
In this example, $\mH = \spn(\ket{S}, \ket{T_0}) \subseteq \mH_A \otimes \mH_B$ holds.
We can easily see that $\ket{S}$ and $\ket{T_0}$ are orthogonal.
If $\ket{+}_\alpha$ and $\ket{-}_\alpha$ are the spin-up and spin-down states of a spin-1/2 particle,
then $\ket{S}$ and $\ket{T_0}$ are, respectively, singlet and triplet states of two particles.
Suppose that Alice and Bob are spatially separated from each other
and share a pair of particles in a state of either $\ket{S}$ or $\ket{T_0}$.
They want to perfectly discriminate between the orthogonal states $\ket{S}$ and $\ket{T_0}$
by one-way LOCC.
This problem is identical to the problem of realizing the projective measurement
$\{ \ket{S}\bra{S}, \ket{T_0}\bra{T_0} \}$ on $\mH$ by one-way LOCC.
If Alice simply performs a measurement in the ONB $\{ \ket{+}_A, \ket{-}_A \}$,
then Bob cannot discriminate between $\ket{S}$ and $\ket{T_0}$;
for example, if the outcome of Alice's measurement is $\ket{+}_A$, then
Bob's state is transformed into $\ket{-}_B$, regardless of whether they share $\ket{S}$ or $\ket{T_0}$.
Thus, Alice needs to use a proper ONB.
$\ket{S}$ and $\ket{T_0}$ are rewritten as
\begin{eqnarray}
 \ket{S} &=& \frac{- \ket{0}_A\ket{1}_B + \ket{1}_A\ket{0}_B}{\sqrt{2}}, \nonumber \\
 \ket{T_0} &=& \frac{\ket{0}_A\ket{0}_B - \ket{1}_A\ket{1}_B}{\sqrt{2}}, \label{eq:Walgate_ex}
\end{eqnarray}
where $\{ \ket{0}_\alpha = (\ket{+}_\alpha + \ket{-}_\alpha)/\sqrt{2}, \ket{1}_\alpha = (\ket{+}_\alpha - \ket{-}_\alpha)/\sqrt{2} \}$
$~(\alpha \in \{ A, B \})$ is the ONB in $\mH_\alpha$.
Alice may just perform a measurement in the ONB $\{ \ket{0}_A, \ket{1}_A \}$ and tell the result to Bob,
and he can then find out which state they share by discriminating between $\ket{0}_B$ and $\ket{1}_B$.

From \cite{Wal-Sho-Har-Ved-2000},
for any 2D $(2,N)$-space, $\mH$, any ONB $\{ \ket{\pi}, \ket{\pi^\perp} \}$ in $\mH$
can be represented as the following form in Alice's proper ONB $\{ \ket{0}_A, \ket{1}_A \}$:
\begin{eqnarray}
 \ket{\pi} &=& \ket{0}_A \ket{\eta_0}_B + \ket{1}_A \ket{\eta_1}_B, \nonumber \\
 \ket{\pi^\perp} &=& \ket{0}_A \ket{\nu_0}_B + \ket{1}_A \ket{\nu_1}_B, \label{eq:pi}
\end{eqnarray}
where $\ket{\eta_k}_B$ and $\ket{\nu_k}_B$ are orthogonal for each $k \in \{ 0, 1 \}$
but not necessarily normalized.
We can see that (\ref{eq:Walgate_ex}) is a special form of (\ref{eq:pi}).
Similar to the above example, the projective measurement
$\{ \ket{\pi}\bra{\pi}, \ket{\pi^\perp}\bra{\pi^\perp} \}$
can be realized by one-way LOCC
if Alice measures her side of the system in the ONB $\{ \ket{0}_A, \ket{1}_A \}$
and Bob discriminates between $\ket{\eta_k}_B$ and $\ket{\nu_k}_B$.

\section{Realization of any 2D measurement by one-way LOCC} \label{sec:nonproj}

Now, we consider realizing a non-projective measurement
$\{ \Pi_m \}_{m=1}^M$ with finite rank-one POVM operators on a 2D $(2,N)$-space $\mH$ by one-way LOCC.
Let us represent $\Pi_1$ as
\begin{eqnarray}
 \Pi_1 &=& \gamma_1 \ket{\pi}\bra{\pi}, \label{eq:Pi1}
\end{eqnarray}
with $0 < \gamma_1 \le 1$ and $\braket{\pi | \pi} = 1$.
Let $\ket{\pi^\perp} \in \mH$ be a normalized vector perpendicular to $\ket{\pi}$
so that $\{ \ket{\pi}, \ket{\pi^\perp} \}$ is an ONB in $\mH$.
We choose an ONB $\{ \ket{0}_A$, $\ket{1}_A \}$ in $\mH_A$
such that $\ket{\pi}$ and $\ket{\pi^\perp}$ are expressed in the form of (\ref{eq:pi}).
Let $\mH_B^{(k)} = \spn(\ket{\eta_k}_B, \ket{\nu_k}_B)$;
then, $\mH_B = \mH_B^{(0)} \cup \mH_B^{(1)}$ obviously holds.
Also, let $P$ be the orthogonal projection operator onto $\mH$
and $I_B$ be the identity operator on $\mH_B$.
Let
\begin{eqnarray}
 \eta_k ~=~ \braket{\eta_k | \eta_k}_B, ~~ \nu_k ~=~ \braket{\nu_k | \nu_k}_B, ~~ k \in \{ 0, 1 \}.
  \label{eq:eta_nu}
\end{eqnarray}
From (\ref{eq:pi}), we have
\begin{eqnarray}
 \eta_0 + \eta_1 &=& \braket{\pi | \pi} ~=~ 1, \nonumber \\
 \nu_0 + \nu_1 &=& \braket{\pi^\perp | \pi^\perp} ~=~ 1. \label{eq:sum_eta_nu}
\end{eqnarray}
Thus, we can assume without loss of generality (by suitably permuting $\ket{0}_A$ and $\ket{1}_A$)
that $\eta_0 \ge \nu_0$.

\subsection{Simple sufficient condition for realization by one-way LOCC} \label{subsec:simple}

In this subsection, we consider the case in which there exist
Bob's measurements $\{ \Phi_m^{(0)} \}_{m=1}^M$ on $\mH_B^{(0)}$
and $\{ \Phi_m^{(1)} \}_{m=1}^M$ on $\mH_B^{(1)}$ such that for any $m$ with $1 \le m \le M$,
$\Pi_m$ is expressed by
\begin{eqnarray}
 \Pi_m &=& P \left( \ket{0}\bra{0}_A \otimes \Phi_m^{(0)} + \ket{1}\bra{1}_A \otimes \Phi_m^{(1)} \right) P.
  \label{eq:Pi_LOCC}
\end{eqnarray}
In this case, $\{ \Pi_m \}$ is realized by one-way LOCC when Alice 
measures her side of the system in the ONB $\{ \ket{0}_A, \ket{1}_A \}$,
as shown in the following lemma.

\begin{lemma} \label{lemma:LOCC_cond}
 Any measurement $\{ \Pi_m \}_{m=1}^M$ with rank-one POVM operators on a 2D $(2,N)$-space
 can be realized by one-way LOCC
 if there exist Bob's measurements $\{ \Phi_m^{(0)} \}_{m=1}^M$ and $\{ \Phi_m^{(1)} \}_{m=1}^M$
 that satisfy (\ref{eq:Pi_LOCC}).
\end{lemma}

\begin{proof}
 We consider the following one-way LOCC measurement:
 Alice measures her side in the ONB $\{ \ket{0}_A, \ket{1}_A \}$
 and reports the result $k \in \{ 0, 1 \}$ to Bob,
 and he then performs a corresponding measurement $\{ \Phi_m^{(k)} \}_{m=1}^M$.
 They regard Bob's result $m$ as the measurement outcome.
 This measurement can obviously be expressed by the POVM
 $\{ \Omega_m \}_{m=1}^M$, where
 \begin{eqnarray}
  \Omega_m &=& \ket{0}\bra{0}_A \otimes \Phi_m^{(0)} + \ket{1}\bra{1}_A \otimes \Phi_m^{(1)}.
 \end{eqnarray}
 From (\ref{eq:Pi_LOCC}), $\Pi_m = P \Omega_m P$ holds for any $m$ with $1 \le m \le M$,
 which means that $\{ \Pi_m \}$ can be realized by one-way LOCC.
 \QED
\end{proof}

We can derive a necessary and sufficient condition that
there exist Bob's measurements $\{ \Phi_m^{(0)} \}_{m=1}^M$ and $\{ \Phi_m^{(1)} \}_{m=1}^M$
that satisfy (\ref{eq:Pi_LOCC}) as given in the following lemma
(proof in Appendix~\ref{append:LOCC}).

\begin{lemma} \label{lemma:LOCC}
 Let $\{ \Pi_m \}_{m=1}^M$ be a 2D measurement with rank-one POVM operators on a 2D $(2,N)$-space.
 A necessary and sufficient condition that
 there exist Bob's measurements $\{ \Phi_m^{(0)} \}_{m=1}^M$ and $\{ \Phi_m^{(1)} \}_{m=1}^M$
 satisfying (\ref{eq:Pi_LOCC})
 is that $\{ c_m \}_{m=1}^M$ exists such that
 \begin{eqnarray}
  && 0 \le c_m \le 1, ~~~ 1 \le m \le M, \nonumber \\
  && \sum_{m=1}^M c_m \Pi_m = Z_0, \label{eq:sum_cm}
 \end{eqnarray}
 where
 \begin{eqnarray}
  Z_0 &=& P (\ket{0}\bra{0}_A \otimes I_B) P. \label{eq:Z0_def}
 \end{eqnarray}
\end{lemma}
In particular, setting $c_m = c_2$ for any $m \ge 3$ in Lemma~\ref{lemma:LOCC}
gives the following proposition (proof in Appendix~\ref{append:LOCC1}).

\begin{proposition} \label{pro:LOCC1}
 Any measurement with finite rank-one POVM operators on a 2D $(2,N)$-space
 can be realized by one-way LOCC if
 \begin{eqnarray}
  \gamma_1 &\ge& \eta_0 - (1 - \gamma_1) \nu_0. \label{eq:LOCC_case1}
 \end{eqnarray}
\end{proposition}
Note that if $\{ \Pi_m \}$ is a projective measurement, then since $\gamma_1 = 1$ holds,
(\ref{eq:LOCC_case1}) always holds.

As an example, we consider $\{ \Pi_m \}_{m=1}^M$ $~(M \ge 3)$ on
$\mH = \spn(\ket{S}, \ket{T_0})$, where $\Pi_m \neq 0$ (i.e., $\gamma_1 > 0$).
For example, a measurement minimizing the average error probability
for the $M$ quantum states $\{ \alpha_m \ket{S} + \beta_m \ket{T_0} \}_{m=1}^M$
with $|\alpha_m|^2 + |\beta_m|^2 = 1$ can often be written as this form.
$\ket{\pi}$ and $\ket{\pi^\perp}$ can be written as
\begin{eqnarray}
 \ket{\pi} &=& x \ket{S} + y \ket{T_0}, \nonumber \\
 \ket{\pi^\perp} &=& - y^* \ket{S} + x^* \ket{T_0}, \label{eq:pi_pip_ex2}
\end{eqnarray}
with some complex values $x$ and $y$ with $|x|^2 + |y|^2 = 1$, where $^*$ denotes the complex conjugate.
Indeed, we can easily verify that $\{ \ket{\pi}, \ket{\pi^\perp} \}$ is an ONB in $\mH$.
Substituting (\ref{eq:Walgate_ex}) into (\ref{eq:pi_pip_ex2}),
we can represent $\ket{\pi}$ and $\ket{\pi^\perp}$ in the form of (\ref{eq:pi}) as
\begin{eqnarray}
 \ket{\pi} &=& \ket{0}_A \frac{y\ket{0}_B - x\ket{1}_B}{\sqrt{2}}
  + \ket{1}_A \frac{x\ket{0}_B - y\ket{1}_B}{\sqrt{2}}, \nonumber \\
 \ket{\pi^\perp} &=& \ket{0}_A \frac{x^*\ket{0}_B - y^*\ket{1}_B}{\sqrt{2}}
  - \ket{1}_A \frac{y^*\ket{0}_B - x^*\ket{1}_B}{\sqrt{2}}.
  \nonumber \\
 \label{eq:pi_ex2}
\end{eqnarray}
From (\ref{eq:pi_ex2}), $\eta_0 = \nu_0 = 1/2$ holds,
and thus (\ref{eq:LOCC_case1}) always holds regardless of $\gamma_1$, $x$, and $y$.
Therefore, from Proposition~\ref{pro:LOCC1},
$\{ \Pi_m \}$ can be realized by one-way LOCC.

Unfortunately, (\ref{eq:LOCC_case1}) does not always hold.
For example, consider the measurement
$\{ \Pi_m = \ket{\pi_m}\bra{\pi_m} \}_{m=1}^3$ on $\mH = \spn(\ket{S}, \ket{T_+})$
with
\begin{eqnarray}
 \ket{\pi_1} &=& \sqrt{\frac{2}{3}}\ket{T_+}, ~
  \ket{\pi_2} = - \sqrt{\frac{1}{6}}\ket{T_+} + \sqrt{\frac{1}{2}}\ket{S}, \nonumber \\
 \ket{\pi_3} &=& - \sqrt{\frac{1}{6}}\ket{T_+} - \sqrt{\frac{1}{2}}\ket{S},
\end{eqnarray}
where $\ket{T_+} = \ket{+}_A\ket{+}_B$.
After some algebra, we have $\eta_0 = 1$, $\nu_0 = 1/2$, and $\gamma_1 = 2/3$,
and thus (\ref{eq:LOCC_case1}) does not hold.
Actually, in this case, (\ref{eq:LOCC_case1}) can be satisfied by
permuting $\Pi_1$ and $\Pi_2$.
However, there exist 2D measurements $\{ \Pi_m \}_{m=1}^M$ such that (\ref{eq:LOCC_case1}) does not hold
for any permutation of the POVM operators.

\subsection{Complete proof of Theorem~\ref{thm:LOCC2}}

From Proposition~\ref{pro:LOCC1},
all we have to do now to prove Theorem~\ref{thm:LOCC2} is to show that
a measurement $\{ \Pi_m \}$ can be realized by one-way LOCC
when (\ref{eq:LOCC_case1}) does not hold.
We here consider making Alice's subsystem interact properly with her auxiliary system.
Let $\mH_S$ be Alice's 2D auxiliary system and $\{ \ket{s_0}, \ket{s_1} \}$ be an ONB in $\mH_S$.
Also, let 
\begin{eqnarray}
 \mL(A) &=& U (\ket{s_0}\bra{s_0} \otimes A) U^\dagger, \nonumber \\
 U &=& U_{SA} \otimes I_B, \label{eq:mL}
\end{eqnarray}
with an operator $A$ on $\mH$,
where $U_{SA}$ is a unitary operator on $\mH_S \otimes \mH_A$.
Also, let $\PE = \mL(P)$, $\rhoE_l = \mL(\rho_l)$, and $\mHE = \spn(\{ \rhoE_l \}_{m=1}^L)$.
We can easily see that $\PE$ is the orthogonal projection operator onto $\mHE$.

We consider the following one-way LOCC measurement:
Alice prepares the auxiliary system in a state $\ket{s_0}$
and transforms $\rho_l$ into $\rhoE_l = \mL(\rho_l)$ using $U_{SA}$.
Then, Alice and Bob perform a measurement $\{ \PiE_m \}_{m=1}^M$,
where $\PiE_m = \mL(\Pi_m)$.
Since $\{ \Pi_m \}$ is on $\mH$, it follows that $\{ \PiE_m \}_{m=1}^M$ is a 2D measurement on $\mHE$.
From (\ref{eq:mL}), for any $l$ with $1 \le l \le L$ and $m$ with $1 \le m \le M$, we have
\begin{eqnarray}
 \Tr(\rhoE_l\PiE_m) &=& \Tr((\ket{s_0}\bra{s_0} \otimes \rho_l) (\ket{s_0}\bra{s_0} \otimes \Pi_m)) \nonumber \\
 &=& \Tr(\rho_l\Pi_m),
\end{eqnarray}
which means that the measurement $\{ \PiE_m \}$ for $\{ \rhoE_l \}$ is intrinsically equivalent to
the measurement $\{ \Pi_m \}$ for $\{ \rho_l \}$.
Thus, to show that $\{ \Pi_m \}$ can be realized by one-way LOCC,
it suffices to find $U_{SA}$ such that $\{ \PiE_m \}$ can be
realized by one-way LOCC.
Note that since $\PiE_m = \mL(\Pi_m)$, for any $m$ with $1 \le m \le M$, $\Pi_m$ can be expressed by
\begin{eqnarray}
 \Pi_m &=& \braket{s_0 | U^\dagger \PiE_m U | s_0}.
\end{eqnarray}

We consider the case in which there exist
measurements $\{ \PhiE{0}_m \}_{m=1}^M$ and $\{ \PhiE{1}_m \}_{m=1}^M$ such that,
for any $m$ with $1 \le m \le M$, $\PiE_m$ is expressed by
\begin{eqnarray}
 \PiE_m &=&
  \PE \left( \ket{s_0}\bra{s_0} \otimes \PhiE{0}_m + \ket{s_1}\bra{s_1} \otimes \PhiE{1}_m \right) \PE.
  \label{eq:Pi_noproj_LOCC}
\end{eqnarray}
The following lemma states that $\{ \Pi_m \}$ can be realized by one-way LOCC
if $\{ \PhiE{0}_m \}$ and $\{ \PhiE{1}_m \}$ can be realized by one-way LOCC.

\begin{lemma} \label{lemma:LOCC_cond_noproj}
 Any measurement $\{ \Pi_m \}_{m=1}^M$ with rank-one POVM operators on a 2D $(2,N)$-space
 can be realized by one-way LOCC
 if a unitary operator $U_{SA}$ on $\mH_S \otimes \mH_A$ exists
 such that there exist measurements $\{ \PhiE{0}_m \}_{m=1}^M$ and $\{ \PhiE{1}_m \}_{m=1}^M$
 that can be realized by one-way LOCC and satisfy (\ref{eq:Pi_noproj_LOCC}),
 where $\PiE_m = \mL(\Pi_m)$.
\end{lemma}

\begin{proof}
 As described above, if $\{ \PiE_m \}$ can be realized by one-way LOCC,
 then $\{ \Pi_m \}$ can also be realized by one-way LOCC.
 We consider the following one-way LOCC measurement for $\{ \rhoE_l \}$ (denoted by $\{ \OmegaE_m \}_{m=1}^M$):
 Alice first performs a measurement on $\mH_S$ in the ONB $\{ \ket{s_0}, \ket{s_1} \}$.
 Let $k \in \{ 0, 1 \}$ be its result.
 Alice and Bob then perform a measurement $\{ \PhiE{k}_m \}$
 and regard its result as the result of $\{ \OmegaE_m \}$.
 $\OmegaE_m$ is obviously expressed by
 \begin{eqnarray}
  \OmegaE_m &=& \ket{s_0}\bra{s_0} \otimes \PhiE{0}_m + \ket{s_1}\bra{s_1} \otimes \PhiE{1}_m.
 \end{eqnarray}
 From (\ref{eq:Pi_noproj_LOCC}), $\PiE_m = \PE \OmegaE_m \PE$ holds for any $m$ with $1 \le m \le M$,
 which means that $\{ \PiE_m \}$ can be realized by one-way LOCC.
 \QED
\end{proof}

Using Lemma~\ref{lemma:LOCC_cond_noproj}, we can show the following proposition.

\begin{proposition} \label{pro:LOCC2}
 Any measurement with finite rank-one POVM operators on a 2D $(2,N)$-space
 can be realized by one-way LOCC if
 \begin{eqnarray}
  \gamma_1 &<& \eta_0 - (1 - \gamma_1) \nu_0. \label{eq:LOCC_case2}
 \end{eqnarray}
\end{proposition}

\begin{proof}
 Let $\{ \Pi_m \}_{m=1}^M$ be a measurement with rank-one POVM operators on a 2D $(2,N)$-space $\mH$.
 Assume that $\{ \Pi_m \}$ satisfies (\ref{eq:LOCC_case2}).
 From Lemma~\ref{lemma:LOCC_cond_noproj}, it suffices to show that
 a unitary operator $U_{SA}$ on $\mH_S \otimes \mH_A$ exists
 such that there exist measurements $\{ \PhiE{0}_m \}_{m=1}^M$ and $\{ \PhiE{1}_m \}_{m=1}^M$
 that can be realized by one-way LOCC and satisfy (\ref{eq:Pi_noproj_LOCC}).

 First, we show a unitary operator $U_{SA}$ and
 measurements $\{ \PhiE{0}_m \}_{m=1}^M$ and $\{ \PhiE{1}_m \}_{m=1}^M$
 that satisfy (\ref{eq:Pi_noproj_LOCC}).
 Also, we show $\PhiE{1}_1 = 0$.
 We choose $U_{SA}$ such that
 \begin{eqnarray}
  U_{SA} \ket{s_0}\ket{0}_A &=& (\sin\theta \ket{s_0} + \cos\theta \ket{s_1}) \ket{0}_A, \nonumber \\
  U_{SA} \ket{s_0}\ket{1}_A &=& \ket{s_1}\ket{1}_A \label{eq:U_SA}
 \end{eqnarray}
 for some real number $\theta$.
 Such $U_{SA}$ is not uniquely determined; we can choose any $U_{SA}$ satisfying (\ref{eq:U_SA}).
 Let
 \begin{eqnarray}
  \ZE_0 &=& \PE (\ket{s_0}\bra{s_0} \otimes I_{AB}) \PE, \label{eq:ZE0}
 \end{eqnarray}
 where $I_{AB}$ is the identity operator on $\mH_A \otimes \mH_B$.
 Using Lemma~\ref{lemma:LOCC} with replacing $\mH_A$ by $\mH_S$,
 $\mH_B$ by $\mH_A \otimes \mH_B$, and $\ket{0}_A$ by $\ket{s_0}$,
 we find that if $\{ \cE_m \}_{m=1}^M$ exists such that
 \begin{eqnarray}
  && 0 \le \cE_m \le 1, ~~~ 1 \le m \le M, \nonumber \\
  && \sum_{m=1}^M \cE_m \PiE_m = \ZE_0, \label{eq:sum_cm_ex}
 \end{eqnarray}
 then there exist POVMs $\{ \PhiE{0}_m \}$ and $\{ \PhiE{1}_m \}$ such that (\ref{eq:Pi_noproj_LOCC}) holds.
 We can show that there exists $\{ \cE_m \}$ such that (\ref{eq:sum_cm_ex}) and $\PhiE{1}_1 = 0$ hold
 if
 \begin{eqnarray}
  \sin^2 \theta &=& \frac{\gamma_1}{\eta_0 - (1 - \gamma_1)\nu_0} \label{eq:sin2alpha}
 \end{eqnarray}
 holds (see Appendix~\ref{append:alpha}).
 Note that from (\ref{eq:LOCC_case2}),
 the right-hand side of (\ref{eq:sin2alpha}) does not exeed 1, and thus there exists $\theta$
 satisfying (\ref{eq:sin2alpha}).

 Next, we show that such measurements $\{ \PhiE{0}_m \}$ and $\{ \PhiE{1}_m \}$ can be realized by one-way LOCC.
 Let $k$ be the outcome of the measurement in the ONB $\{ \ket{s_0}, \ket{s_1} \}$.
 If $k = 0$, then, from (\ref{eq:U_SA}), the state of $\mH_A$ is always projected onto $\ket{0}_A$,
 which indicates that $\PhiE{0}_m$ can be written in the form
 \begin{eqnarray}
  \PhiE{0}_m = \ket{0}\bra{0}_A \otimes \Psi_m, \label{eq:Psim}
 \end{eqnarray}
 where $\{ \Psi_m \}_{m=1}^M$ is a POVM on Bob's side of the system.
 Thus, in this case, Bob may simply perform the measurement $\{ \Psi_m \}$.
 If $k = 1$, then Alice and Bob have to perform the 2D measurement $\{ \PhiE{1}_m \}_{m=1}^M$.
 Since $\PhiE{1}_1 = 0$ holds,
 $\{ \PhiE{1}_m \}$ has less than $M$ non-zero POVM operators.
 Therefore, the problem of realizing $\{ \Pi_m \}$ with $M$ POVM operators by one-way LOCC
 is reduced to the problem of realizing $\{ \PhiE{1}_1 \}$ with $M'$ POVM operators
 by one-way LOCC, where $M' < M$.
 Therefore, by iteratively performing the procedure stated in this paper,
 $\{ \PhiE{1}_m \}$ can be realized by one-way LOCC,
 since any 2D measurement with less than three non-zero POVM operators can obviously
 be realized by one-way LOCC \cite{Wal-Sho-Har-Ved-2000}.
 \QED
\end{proof}

\begin{proofa}
 Obvious from Propositions~\ref{pro:LOCC1} and \ref{pro:LOCC2}.
 \QED
\end{proofa}

\subsection{Schematic diagram for realizing 2D measurement}

\begin{figure}[bt]
 \centering
 \includegraphics[scale=0.8]{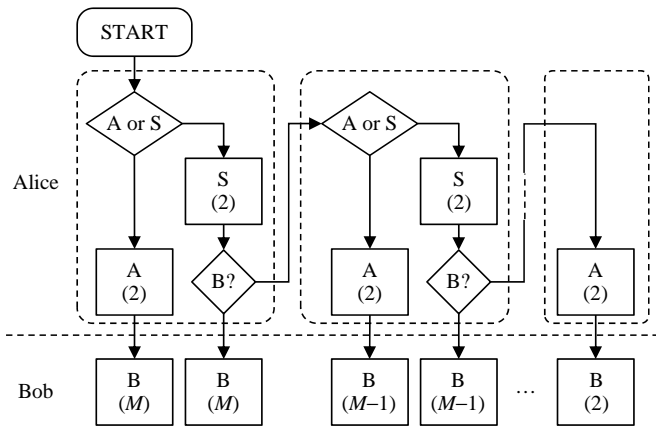}
 \caption{\label{fig:LOCC}A schematic diagram for realizing a measurement
 with finite rank-one POVM operators on a 2D $(2,N)$-space $\mH$ by one-way LOCC.
 Diamonds represent decisions.
 Rectangles represent measurements on $\mH_S$, $\mH_A$, or $\mH_B^{(k)}$ $~(k \in \{ 0, 1 \})$.
 Each measurement on $\mH_S$ is performed after 
 Alice's state interacts with her auxiliary system.
 Values in the brackets show the number of measurement outcomes.
 }
\end{figure}

A schematic diagram of our measurement process
in the case of a measurement with finite rank-one POVM operators on a 2D $(2,N)$-space
is sketched in Fig.~\ref{fig:LOCC}.
Alice first determines whether she performs a binary measurement on $\mH_A$ or
makes her system interact with her auxiliary system $\mH_S$ followed by
performing a binary measurement on $\mH_S$.
The decision rule is given by (\ref{eq:LOCC_case1}).
Then, in the former case, Alice tells the result $k$ to Bob,
and he performs a measurement on $\mH_B^{(k)}$.
In the latter case, whether Alice or Bob performs
a measurement is determined by the result of Alice's measurement in the ONB
$\{ \ket{s_0}, \ket{s_1} \}$.
Alice repeats the above sequence the necessary number of times.
This procedure stops after a finite number of steps.
Bob may perform a measurement only once at an appropriate time.

The entire algorithm for realizing such a measurement is found in the following pseudocode:
\begin{algorithmic}[1]
 \STATE {\bfseries Input:} a quantum state $\rho_l$ and a POVM $\{ \Pi_m \}_{m=1}^M$ with finite rank-one POVM operators on a 2D $(2,N)$-space.
 \REPEAT
 \STATE Compute $\gamma_1$, $\nu_0$, and $\eta_0$ from (\ref{eq:Pi1}) and (\ref{eq:eta_nu}).
 \STATE Compute $\ket{0}_A$ and $\ket{1}_A$ such that (\ref{eq:pi}) holds.
 \IF{(\ref{eq:LOCC_case1}) holds}
 \STATE Alice performs a measurement in the ONB $\{ \ket{0}_A, \ket{1}_A \}$ and reports her result $k \in \{ 0, 1 \}$ to Bob.
 \STATE Bob performs a measurement $\{ \Phi_m^{(k)} \}_{m=1}^M$ ($\Phi_m^{(k)}$ is obtained from (\ref{eq:Phi})).
 \ELSE
 \STATE Compute $U_{SA}$ such that (\ref{eq:U_SA}) holds ($\theta$ is obtained from (\ref{eq:sin2alpha})).
 \STATE Alice prepares the auxiliary system in a state $\ket{s_0}$ and transforms $\rho_l$ into $\rhoE_l = \mL(\rho_l)$.
 \STATE Alice performs a measurement in the ONB $\{ \ket{s_0}, \ket{s_1} \}$ (denote its result as $k$).
 \IF{$k = 0$}
 \STATE Bob performs a measurement $\{ \Psi_m \}_{m=1}^M$ satisfying (\ref{eq:Psim}).
 \ELSE
 \STATE Regard $\rhoE_l$ and $\{ \PhiE{1}_m \}_{m=1}^M$ as $\rho_l$ and $\{ \Pi_m \}_{m=1}^M$, respectively.
 \ENDIF
 \ENDIF
 \UNTIL{Bob performs a measurement.}
 \STATE {\bfseries Output:} the outcome of Bob's measurement.
\end{algorithmic}

\section{Conclusion} \label{sec:conclusion}

In conclusion, we have proved that any 2D measurement in finite-dimensional multipartite systems
can be realized by one-way LOCC.
This implies that multipartite quantum states on a 2D Hilbert space can always be optimally
distinguished by one-way LOCC no matter which optimality criterion is applied.
This also means that in a 2D case,
any entangled information of quantum states obtained by a global measurement
can also be obtained only by one-way LOCC, at least in finite-dimensional systems.

\appendix

\section{Proof of Lemma~\ref{lemma:LOCC}} \label{append:LOCC}

\subsection{Preparations}

 First, we define some operators.
 Let
 \begin{eqnarray}
  S_k &=& \ket{\pi}\bra{\eta_k}_B + \ket{\pi^\perp}\bra{\nu_k}_B, ~~~ k \in \{ 0, 1 \}, \label{eq:Sk} \\
  T_k &=& \eta_k^- \ket{\eta_k}_B\bra{\pi} + \nu_k^- \ket{\nu_k}_B\bra{\pi^\perp}, ~~ k \in \{ 0, 1 \}, \label{eq:Tk}
 \end{eqnarray}
 where $x^-$ is defined as $x^{-1}$ if $x \neq 0$ and zero otherwise.
 $S_k$ and $T_k$ are operators from $\mH_B^{(k)}$ to $\mH$ and from $\mH$ to $\mH_B^{(k)}$, respectively.
 Let $P_k = S_k T_k$; then, from (\ref{eq:Sk}) and (\ref{eq:Tk}), for any $k \in \{ 0, 1 \}$, we have
 \begin{eqnarray}
  P_k &=& \eta_k \eta_k^- \ket{\pi}\bra{\pi} + \nu_k \nu_k^- \ket{\pi^\perp}\bra{\pi^\perp}. \label{eq:Pk}
 \end{eqnarray}
 Since $\eta_k \eta_k^-$ and $\nu_k \nu_k^-$ are 0 or 1, $P_k$ is the orthogonal projection operator
 onto $\spn(\eta_k \ket{\pi}, \nu_k \ket{\pi^\perp})$
 (note that if $\eta_k \neq 0$ and $\nu_k \neq 0$, then $P_k = P$).
 Also, for any $k \in \{ 0, 1 \}$, we have
 \begin{eqnarray}
  T_k S_k &=& \eta_k^- \ket{\eta_k}\bra{\eta_k}_B + \nu_k^- \ket{\nu_k}\bra{\nu_k}_B ~=~ I_B^{(k)},
   \label{eq:TkSk}
 \end{eqnarray}
 where the second equality follows since $\ket{\eta_k}$ and $\ket{\nu_k}$ are orthogonal vectors
 of $\mH_B^{(k)}$, and (\ref{eq:eta_nu}) holds.
 Moreover, for any operator $X$ on $\mH_B$ and $k \in \{ 0, 1 \}$, we have
 \begin{eqnarray}
  P (\ket{k}\bra{k}_A \otimes X) P &=& S_k X S_k^\dagger. \label{eq:PkXP}
 \end{eqnarray}
 Indeed, from $P = \ket{\pi}\bra{\pi} + \ket{\pi^\perp}\bra{\pi^\perp}$,
 we have
 \begin{eqnarray}
  P (\ket{k}_A \otimes I_B) &=& (\ket{\pi}\bra{\pi} + \ket{\pi^\perp}\bra{\pi^\perp}) (\ket{k}_A \otimes I_B)
   \nonumber \\
  &=& S_k,
 \end{eqnarray}
 where the second line follows from (\ref{eq:pi}).
 Thus, since $\ket{k}\bra{k}_A \otimes X = (\ket{k}_A \otimes I_B) X (\bra{k}_A \otimes I_B)$,
 (\ref{eq:PkXP}) holds.

 We also define
 \begin{eqnarray}
  Z_k &=& P (\ket{k}\bra{k}_A \otimes I_B) P, ~~~ k \in \{ 0, 1 \}, \label{eq:Zk_def}
 \end{eqnarray}
 which includes the definition of $Z_0$ in (\ref{eq:Z0_def}).
 We can easily obtain $S_k I_B = S_k$ from (\ref{eq:Sk});
 thus, from (\ref{eq:PkXP}), we have
 \begin{eqnarray}
  Z_k &=& S_k I_B S_k^\dagger ~=~ S_k S_k^\dagger. \label{eq:ZkPPP}
 \end{eqnarray}
 Substituting (\ref{eq:Sk}) into (\ref{eq:ZkPPP}) yields
 \begin{eqnarray}
  Z_k &=& \eta_k \ket{\pi}\bra{\pi} + \nu_k \ket{\pi^\perp}\bra{\pi^\perp}. \label{eq:Zk}
 \end{eqnarray}

\subsection{Necessity}

 Here, we prove the necessity.
 Since $\Pi_m$ is a rank-one operator, to satisfy (\ref{eq:Pi_LOCC}),
 there must exist $\{ c_m \}_{m=1}^M$ with $0 \le c_m \le 1$ such that for any $m$ with $1 \le m \le M$,
 \begin{eqnarray}
  P \left( \ket{0}\bra{0}_A \otimes \Phi_m^{(0)} \right) P &=& c_m \Pi_m, \nonumber \\
  P \left( \ket{1}\bra{1}_A \otimes \Phi_m^{(1)} \right) P &=& (1 - c_m) \Pi_m.
   \label{eq:Pi_LOCC_decomp}
 \end{eqnarray}
 In contrast, since $\{ \Phi_m^{(0)} \}$ is a POVM on $\mH_B^{(0)}$, $\sum_{m=1}^M \Phi_m^{(0)} = I_B^{(0)}$ holds,
 where $I_B^{(k)}$ is the identity operator on $\mH_B^{(k)}$.
 Thus, from (\ref{eq:Pi_LOCC_decomp}), we have
 \begin{eqnarray}
  \sum_{m=1}^M c_m \Pi_m &=& P \left( \ket{0}\bra{0}_A \otimes \sum_{m=1}^M \Phi_m^{(0)} \right) P \nonumber \\
  &=& P ( \ket{0}\bra{0}_A \otimes I_B^{(0)} ) P \nonumber \\
  &=& S_0 I_B^{(0)} S_0^\dagger \nonumber \\
  &=& S_0 S_0^\dagger \nonumber \\
  &=& Z_0,
 \end{eqnarray}
 where the third and fifth lines follow from (\ref{eq:PkXP}) and (\ref{eq:ZkPPP}), respectively.
 Therefore, $\{ c_m \}$ satisfies (\ref{eq:sum_cm}).

\subsection{Sufficiency}

 Here, we prove the sufficiency.
 Assume that there exists $\{ c_m \}_{m=1}^M$ satisfying (\ref{eq:sum_cm}).
 It is sufficient to show that POVMs $\{ \Phi_m^{(0)} \}$ and $\{ \Phi_m^{(1)} \}$ exist
 such that (\ref{eq:Pi_LOCC_decomp}) holds.
 Indeed, in this case (\ref{eq:Pi_LOCC}) is obtained from
 the sum of the first and second lines of (\ref{eq:Pi_LOCC_decomp}).
 Let
 \begin{eqnarray}
  \Phi_m^{(k)} &=& c_m^{(k)} T_k \Pi_m T_k^\dagger, ~~~ 1 \le m \le M, ~ k \in \{ 0, 1 \}, \nonumber \\
  c_m^{(0)} &=& c_m, ~~~~~~~~~~~~~~~ 1 \le m \le M, \nonumber \\
  c_m^{(1)} &=& 1 - c_m, ~~~~~~~~~~ 1 \le m \le M. \label{eq:Phi}
 \end{eqnarray}
 $\Phi_m^{(k)}$ is obviously a positive semidefinite operator on $\mH_B^{(k)}$.
 We show that $\{ \Phi_m^{(0)} \}$ and $\{ \Phi_m^{(1)} \}$ are POVMs satisfying
 (\ref{eq:Pi_LOCC_decomp}).
 Since $Z_0 + Z_1 = P$ holds from (\ref{eq:sum_eta_nu}) and (\ref{eq:Zk}),
 $\sum_{m=1}^M (1 - c_m)\Pi_m = P - Z_0 = Z_1$ holds from (\ref{eq:sum_cm}),
 which gives
 \begin{eqnarray}
  \sum_{m=1}^M c_m^{(k)}\Pi_m = Z_k, ~~~ k \in \{ 0, 1 \}. \label{eq:Z01}
 \end{eqnarray}
 Thus, from (\ref{eq:Phi}), for any $k \in \{ 0, 1 \}$, we have
 \begin{eqnarray}
  \sum_{m=1}^M \Phi_m^{(k)} &=& T_k Z_k T_k^\dagger 
   ~=~ T_k S_k S_k^\dagger T_k^\dagger
   ~=~ I_B^{(k)},
 \end{eqnarray}
 where the second and third equalities follow from (\ref{eq:ZkPPP}) and (\ref{eq:TkSk}), respectively.
 Therefore, $\{ \Phi_m^{(0)} \}$ and $\{ \Phi_m^{(1)} \}$ are POVMs.
 From (\ref{eq:PkXP}) and (\ref{eq:Phi}), for any $k \in \{ 0, 1 \}$, we have
 \begin{eqnarray}
  P \left( \ket{k}\bra{k}_A \otimes \Phi_m^{(k)} \right) P &=&
   S_k \Phi_m^{(k)} S_k^\dagger \nonumber \\
  &=& c_m^{(k)} S_k T_k \Pi_m T_k^\dagger S_k^\dagger \nonumber \\
  &=& c_m^{(k)} P_k \Pi_m P_k.
 \end{eqnarray}
 Thus, to prove (\ref{eq:Pi_LOCC_decomp}), it suffices to show $c_m^{(k)} P_k \Pi_m P_k = c_m^{(k)} \Pi_m$.
 Since $P_k \ge Z_k$ holds from (\ref{eq:Pk}) and (\ref{eq:Zk})
 ($A \ge B$ denotes that $A - B$ is positive semi-definite), we have
 \begin{eqnarray}
  P_k \ge Z_k \ge c_m^{(k)} \Pi_m, ~~~ k \in \{ 0, 1 \},
 \end{eqnarray}
 where the second inequality follows from (\ref{eq:Z01}).
 Thus, since $P_k$ is the orthogonal projection operator, $c_m^{(k)} P_k \Pi_m P_k = c_m^{(k)} \Pi_m$ holds.
 Therefore, (\ref{eq:Pi_LOCC_decomp}) holds.
 \QED

\section{Proof of Proposition~\ref{pro:LOCC1}} \label{append:LOCC1}

 Let $\{ \Pi_m \}_{m=1}^M$ be a measurement with rank-one POVM operators on a 2D $(2,N)$-space $\mH$.
 Assume that $\{ \Pi_m \}$ satisfies (\ref{eq:LOCC_case1}).
 From Lemmas~\ref{lemma:LOCC_cond} and \ref{lemma:LOCC},
 it suffices to show that there exists $\{ c_m \}_{m=1}^M$ satisfying (\ref{eq:sum_cm}).
 Let
 \begin{eqnarray}
  c_1 &=& \frac{\eta_0 - (1 - \gamma_1)\nu_0}{\gamma_1}, \nonumber \\
  c_m &=& \nu_0, ~~~ m \in \{ 2, 3, \cdots, M \}. \label{eq:c12}
 \end{eqnarray}
 We can see that $0 \le c_m \le 1$ for any $m$ with $1 \le m \le M$.
 Indeed, since $0 \le \nu_0 \le 1$, $0 \le c_m \le 1$ holds for any $m \ge 2$.
 Since $\eta_0 \ge \nu_0 \ge (1 - \gamma_1)\nu_0$, which follows from $\gamma_1 > 0$, $c_1 \ge 0$ holds.
 Moreover, $c_1 \le 1$ holds from (\ref{eq:LOCC_case1}).
 From (\ref{eq:c12}), we obtain
 \begin{eqnarray}
  \sum_{m=1}^M c_m \Pi_m &=& c_1 \Pi_1 + \nu_0 (P - \Pi_1) \nonumber \\
  &=& (c_1 - \nu_0) \Pi_1 + \nu_0 P \nonumber \\
  &=& (\eta_0 - \nu_0) \ket{\pi}\bra{\pi} + \nu_0 (\ket{\pi}\bra{\pi} + \ket{\pi^\perp}\bra{\pi^\perp}) \nonumber \\
  &=& \eta_0 \ket{\pi}\bra{\pi} + \nu_0 \ket{\pi^\perp}\bra{\pi^\perp} \nonumber \\
  &=& Z_0,
 \end{eqnarray}
 where the third line follows from $\Pi_1 = \gamma_1 \ket{\pi}\bra{\pi}$
 and $P = \ket{\pi}\bra{\pi} + \ket{\pi^\perp}\bra{\pi^\perp}$,
 and the last line follows from (\ref{eq:Zk}).
 Therefore, $\{ c_m \}$ of (\ref{eq:c12}) satisfies (\ref{eq:sum_cm}).
 \QED

\section{Supplement of (\ref{eq:sum_cm_ex}) and (\ref{eq:sin2alpha})} \label{append:alpha}

 Assume (\ref{eq:sin2alpha}); we will show that there exists $\{ \cE_m \}$ such that
 (\ref{eq:sum_cm_ex}) and $\PhiE{1}_1 = 0$ hold.

 In preparation, we show that (\ref{eq:sum_cm_ex}) is equivalent to
 \begin{eqnarray}
  \sum_{m=1}^M \cE_m \Pi_m = (\sin^2 \theta) Z_0. \label{eq:sum_cm_ex2}
 \end{eqnarray}
 Premultiplying and postmultiplying both sides of (\ref{eq:sum_cm_ex}) by $U^\dagger$ and $U$,
 respectively, yield
 \begin{eqnarray}
  \ket{s_0}\bra{s_0} \otimes \sum_{m=1}^M \cE_m \Pi_m &=& U^\dagger \ZE_0 U. \label{eq:U_sum_cm_ex_U}
 \end{eqnarray}
 Let $P_{s_0} = \ket{s_0}\bra{s_0} \otimes I_{AB}$; then, from (\ref{eq:ZE0}),
 $\ZE_0 = \PE P_{s_0} \PE$ holds.
 Thus, we have
 \begin{eqnarray}
  U^\dagger \ZE_0 U &=& U^\dagger \PE P_{s_0} \PE U ~=~ D^\dagger D, \label{eq:UZU}
 \end{eqnarray}
 where
 \begin{eqnarray}
  D &=& P_{s_0} \PE U.
 \end{eqnarray}
 The second equation of (\ref{eq:UZU}) follows from $P_{s_0} = P_{s_0}^2$.
 In contrast, from $\PE = \mL(P) = U (\ket{s_0}\bra{s_0} \otimes P) U^\dagger$,
 (\ref{eq:mL}), and (\ref{eq:U_SA}), we have
 \begin{eqnarray}
  D &=& P_{s_0} U (\ket{s_0}\bra{s_0} \otimes P) \nonumber \\
  &=& P_{s_0} (\sin\theta \ket{s_0}\bra{s_0} \otimes \ket{0}\bra{0}_A \otimes I_B)
   (\ket{s_0}\bra{s_0} \otimes P) \nonumber \\
  &=& \sin\theta \ket{s_0}\bra{s_0} \otimes (I_{AB} (\ket{0}\bra{0}_A \otimes I_B) P) \nonumber \\
  &=& \sin\theta \ket{s_0}\bra{s_0} \otimes ((\ket{0}\bra{0}_A \otimes I_B) P). \label{eq:D}
 \end{eqnarray}
 (\ref{eq:UZU}) and (\ref{eq:D}) yield
 \begin{eqnarray}
  U^\dagger \ZE_0 U &=& \sin^2\theta \ket{s_0}\bra{s_0} \otimes (P (\ket{0}\bra{0}_A \otimes I_B) P) \nonumber \\
  &=& \sin^2\theta \ket{s_0}\bra{s_0} \otimes Z_0, \label{eq:UZU2}
 \end{eqnarray}
 where the second line follows from (\ref{eq:Zk_def}).
 From (\ref{eq:U_sum_cm_ex_U}) and (\ref{eq:UZU2}), (\ref{eq:sum_cm_ex}) is equivalent to
 (\ref{eq:sum_cm_ex2}).

 Now, we show that there exists $\{ \cE_m \}$ such that
 (\ref{eq:sum_cm_ex}) and $\PhiE{1}_1 = 0$ hold.
 Let $\cE_1 = 1$ and $\cE_m = \cE_2$ for any $m \ge 3$.
 As shown in the proof of Lemma~\ref{lemma:LOCC},
 $c_m = 1$ (i.e., $c_m^{(1)} = 0$) yields $\Phi_m^{(1)} = 0$ from (\ref{eq:Phi}),
 which indicates that $\PhiE{1}_1 = 0$ holds from $\cE_1 = 1$.
 We have
 \begin{eqnarray}
  \sum_{m=1}^M \cE_m \Pi_m &=& \Pi_1 + \cE_2 (P - \Pi_1) \nonumber \\
  &=& (1 - \cE_2) \Pi_1 + \cE_2 P \nonumber \\
  &=& (\gamma_1 + (1 - \gamma_1) \cE_2) \ket{\pi}\bra{\pi} + \cE_2 \ket{\pi^\perp}\bra{\pi^\perp}, \nonumber \\
 \end{eqnarray}
 where the last line follows from $\Pi_1 = \gamma_1 \ket{\pi}\bra{\pi}$
 and $P = \ket{\pi}\bra{\pi} + \ket{\pi^\perp}\bra{\pi^\perp}$.
 Thus, from (\ref{eq:Zk}), (\ref{eq:sum_cm_ex2}) (i.e., (\ref{eq:sum_cm_ex})) is equivalent to
 \begin{eqnarray}
  \gamma_1 + (1 - \gamma_1) \cE_2 &=& \eta_0 \sin^2\theta, \nonumber \\
  \cE_2 &=& \nu_0 \sin^2\theta, \label{eq:sum_cm_ex3}
 \end{eqnarray}
 so we let $\cE_2 = \nu_0 \sin^2\theta$.
 $0 \le \cE_2 \le 1$ obviously holds.
 We can see from (\ref{eq:sin2alpha}) that (\ref{eq:sum_cm_ex3}) holds;
 therefore, (\ref{eq:sum_cm_ex}) holds.

\begin{acknowledgments}
 We thank O. Hirota and K. Kato of Tamagawa University for the useful discussions we had with them.
 T. S. U. was supported (in part) by JSPS KAKENHI (Grant No. 24360151).
\end{acknowledgments}


\end{document}